\documentclass{llncs}
\usepackage[utf8]{inputenc}
\usepackage{amsfonts,amssymb,amscd,amsmath,enumerate,color}
\usepackage{algorithm}
\usepackage{algcompatible}
\usepackage{algpseudocode}
\usepackage{color}
\usepackage{graphicx}

\newcommand{\Bton}{B[2,n]}

\def\opn#1#2{\def#1{\operatorname{#2}}} % to make operators
\opn\chara{char} \opn\length{\ell} \opn\pd{pd} \opn\rk{rk}
\opn\projdim{proj\,dim} \opn\injdim{inj\,dim} \opn\rank{rank}
\opn\depth{depth} \opn\grade{grade} \opn\height{height}
\opn\embdim{emb\,dim} \opn\codim{codim}

\opn\Tr{Tr} \opn\bigrank{big\,rank}
\opn\superheight{superheight}\opn\lcm{lcm}
\opn\trdeg{tr\,deg}%\emph{
\opn\reg{reg} \opn\lreg{lreg} \opn\ini{in} \opn\lpd{lpd}
\opn\size{size}\opn\bigsize{bigsize}
\opn\cosize{cosize}\opn\bigcosize{bigcosize}
\opn\sdepth{sdepth}\opn\sreg{sreg}
\opn\link{link}\opn\fdepth{fdepth}
\opn\div{div} \opn\Div{Div} \opn\cl{cl} \opn\Cl{Cl}
\opn\Spec{Spec} \opn\Supp{Supp} \opn\supp{supp} \opn\Sing{Sing}
\opn\Ass{Ass} \opn\Min{Min}\opn\Mon{Mon} \opn\dstab{dstab} \opn\astab{astab}
\opn\Syz{Syz}
\opn\Ann{Ann} \opn\Rad{Rad} \opn\Soc{Soc}
\opn\Im{Im} \opn\Ker{Ker} \opn\Coker{Coker} \opn\Am{Am}
\opn\Hom{Hom} \opn\Tor{Tor} \opn\Ext{Ext} \opn\End{End}
\opn\Aut{Aut} \opn\id{id}
\opn\resideal{I(Res)}

\opn\nat{nat}
\opn\pff{pf}%   \pf exists already
\opn\Pf{Pf} \opn\GL{GL} \opn\SL{SL} \opn\mod{mod} \opn\ord{ord}
\opn\Gin{Gin} \opn\Hilb{Hilb}\opn\sort{sort}
\opn\initial{init}
\opn\ende{end}
\opn\height{height}
\opn\type{type}
\opn\Res{Res}
\opn\proj{Proj}
\opn\syl{Syl}
\opn\resideal{I(Res)}
\opn\lm{LM}
\opn\pil{\rightarrow}

\opn\aff{aff} \opn\con{conv} \opn\relint{relint} \opn\st{st}
\opn\lk{lk} \opn\cn{cn} \opn\core{core} \opn\vol{vol}
\opn\link{link} \opn\star{star}\opn\lex{lex}
\opn\gr{gr}

\def\pot#1#2{#1[\kern-0.28ex[#2]\kern-0.28ex]}

\opn\dirlim{\underrightarrow{\lim}}
\opn\inivlim{\underleftarrow{\lim}}
\let\to=\rightarrow

\def\Implies{\ifmmode\Longrightarrow \else
        \unskip${}\Longrightarrow{}$\ignorespaces\fi}
\def\implies{\ifmmode\Rightarrow \else
        \unskip${}\Rightarrow{}$\ignorespaces\fi}
\def\iff{\ifmmode\Longleftrightarrow \else
        \unskip${}\Longleftrightarrow{}$\ignorespaces\fi}
\let\gets=\leftarrow

\let\:=\colon
 \newtheorem{Theorem}{Theorem}%[section]
 \newtheorem{Lemma}[Theorem]{Lemma}
 \newtheorem{Corollary}[Theorem]{Corollary}

 \newtheorem{Definition}[Theorem]{Definition}
 \newtheorem{Remark}[Theorem]{Remark}

\let\epsilon\varepsilon
\let\kappa=\varkappa
\def\qed{\ifhmode\textqed\fi
      \ifmmode\ifinner\quad\qedsymbol\else\dispqed\fi\fi}
\def\textqed{\unskip\nobreak\penalty50
       \hskip2em\hbox{}\nobreak\hfil\qedsymbol
       \parfillskip=0pt \finalhyphendemerits=0}
\def\dispqed{\rlap{\qquad\qedsymbol}}

\opn\dis{dis}
\def\pnt{{\raise0.5mm\hbox{\large\bf.}}}

\opn\Lex{Lex}
\newcommand{\Co}{\text{Co}}

\newcommand{\sus}{\subseteq}

\title{Eliminating Variables in Boolean Equation Systems}
\author{Bj{\o}rn M{\o}ller Greve\inst{1,2} \and H{\aa}vard Raddum\inst{2} \and Gunnar Fl{\o}ystad\inst{3} \and {\O}yvind Ytrehus\inst{2}}
\institute{Norwegian Defence Research Establishment \and Simula@UiB \and Dept. of Mathematics, UiB}

\date{\today}

\begin{document}

\maketitle

\begin{abstract}
Systems of Boolean equations of low degree arise in a natural way when analyzing block ciphers. The cipher's round functions relate the secret key to auxiliary variables that are introduced by each successive round. In algebraic cryptanalysis, the attacker attempts to solve the resulting equation system in order to extract the secret key. In this paper we study algorithms for eliminating the auxiliary variables from these systems of Boolean equations. It is known that elimination of variables in general increases the degree of the equations involved. In order to contain computational complexity and storage complexity, we present two new algorithms for performing elimination while bounding the degree at $3$, which is the lowest possible for elimination. Further we show that the new algorithms are related to the well known \emph{XL} algorithm. We apply the algorithms to a downscaled version of the LowMC cipher and to a toy cipher based on the Prince cipher, and report on experimental results pertaining to these examples. 
\end{abstract}

\section{Introduction}

A block cipher encryption algorithm $E_K(P)=C$ takes a fixed length plaintext $P$ and a secret key  $K$ as inputs, and produces a ciphertext $C$.  The encryption usually consists of iterating a round function, which in turn is made up by suitable linear and nonlinear transformations.  In a \emph{known plaintext attack}, both $P$ and $C$ are known to the cryptanalyst, who wants to find the secret key $K$. 

Ciphers defined over $GF(2)$ can be described as a system of multivariate Boolean equations of degree $2$. These equations relate the bits of the secret key $K$ and new \emph{auxiliary variables} that arise due to the round functions via the known $P$ and $C$. Solving this system of equations with respect to the secret key $K$ is known as \emph{algebraic cryptanalysis}. 
%
%*** NB : Resten av Introduction bør forenkles. Trenger ikke ta med alle detaljer enda. ***
%
The approach in this paper is to %investigate ways to 
iteratively eliminate the %``many'' 
auxiliary variables that arise in an initial such system of  Boolean equations. % system of degree $2$ describing a block cipher, 
At each iteration, the elimination step converts a current system of polynomial equations $F$ in the variables $x_1,\ldots,x_n$ into a set of new equations $F'$ in $x_2,\ldots,x_n$, so that the solution set of $F'$ is the projection of the solution set of $F$ onto $x_2,\ldots,x_n$. %, and so that the degree of   polynomials in $F'$ is bounded by some value $\delta$. %Moreover we try to eliminate all the auxiliary variables in the equation system, in order to end up with {\it some} low-degree equations in only variables of the bits of the secret key $K$.
%Our contribution is a general framework for performing elimination in the Boolean ring, which not only gives a deeper mathematical insight into the concept of elimination of systems of quadratic equations over the Boolean ring, but also shows how this is connected to previous work  in algebraic cryptanalysis.%, such as the XL algorithm. 
%In this paper w
We propose two algorithms for the elimination step: $L-Elim{\bf *}()$ (a  variant of XL with bounded degree) and $eliminate{\bf *}()$ (a new algorithm based on the mathematical framework developed in this paper), where ${\bf *} \in \{{\bf A},{\bf B}\}$ denotes one of two variants.  In both algorithms the polynomial degree is limited to 3 at all times in order to contain computational complexity and storage complexity.  %The purpose of these algorithms is to apply them to systems of equations describing block ciphers, and then try to find polynomials of degree $3$ or less in only key variables. This degree restriction is in order to limit the computational complexity, which enables us to run experiments on an ordinary computer. 
The highlight of the paper is Theorem \ref{pro:FLG}, where we show that the more efficient algorithm $eliminate{\bf A}()$ produces the same output as $L-Elim{\bf A}()$. In addition we  show that by applying a few extra tricks, the ${\bf B}$-variants of the algorithms can produce more equations than the ${\bf A}$-variants.

The paper is structured as follows. Section~\ref{Notation} describes the problem, the notation, and previous results. The foundation for elimination techniques over the Boolean ring is developed in Section~\ref{sec:elimination}. The new algorithms, $L-Elim{\bf *}()$ and $eliminate{\bf *}()$, are presented in Section~\ref{Algorithms} along with discussions of complexity and information loss.  In Section~\ref{sec:experiments} we report on experimental results of the new algorithms when applied to reduced block ciphers: A reduced version of the LowMC cipher, and a \emph{toy cipher} based on the PRINCE cipher.

%For the reduced LowMC cipher, the algorithms find degree $3$ polynomials and hence breaks the cipher, up to a certain number of rounds. In fact, the new algorithms finds \emph{linear} polynomials in only key variables much quicker than expected, when repeated over several plain/cipher-text pairs encrypted with the same key $K$. However, the algorithms fail to find any degree $3$ polynomials when the number of rounds is increased.

%The algorithm \emph{fails} to find degree $3$ polynomials in only key variables when applied to 4 rounds of the \emph{toy cipher}. In Subsection~\ref{subsec:enformation}, we measure how fast information about the actual solutions is lost due to degree restrictions. The experiments on the toy cipher indicate that there is no loss of information on the correct key(s) for many iterations, but then all information is lost abruptly. 

\section{Notation and Preliminaries}\label{Notation}

Consider the quotient ring of Boolean polynomials in $n$ variables. We denote the ring by 
$$B[1,n] = \mathbb{F}_2 [x_1 ,\ldots ,x_n ]/(x^2_i + x_i|i=1,\ldots ,n).$$
A monomial is a product  $x_{i_1} \cdots x_{i_{\delta}}$ of $\delta$ distinct (because $x^2 = x$) variables, where $\delta$ is the degree of this monomial. The degree of a polynomial \[p = \sum_s m_s\]where the $m_s$'s are distinct monomials, is the maximum degree over the monomials in $p$. 
Given a set of polynomial equations $F = \{f_i(x_1,\ldots,x_n) = 0 | i = 1,\dots,m \}$, our objective is to find its set of solutions in the space $\mathbb{F}^n_2$.  The approach we take in this paper is to solve the system of equations by eliminating variables. 

In the following we assume without loss of generality %for the sake of clarity 
that we eliminate variables in the order $x_1, x_2, \ldots ,x_n$. %This is without loss of generality since we can eliminate in any order by renaming variables.  %We order the variables as $x_1 > x_2 > \ldots > x_n$. Note that since $x_i^2 = x_i$ we cannot have a monomial order in the proper sense. The reason is that we are not guaranteed that when monomials $m_1 > m_2$, then $x_im_1 > x_im_2$ for any variable $x_i$. We overcome this problem by requiring that the order on the monomials have the property that if $m_1 > m_2$ and $x_i$ is not a factor of $m_1$, then $x_im_1 > x_im_2$. Examples of such orderings are the lex order on monomials, or the deglex order, where monomials are first partially ordered by degree and then refined by using the lexicographic order in each degree. Another example is the order introduced in Algorithm \ref{alg:splitx1}.
%We order the monomials by the lexicographic order as default.  For the monomials of degree $\leq 2$, the first monomials are 
%\[
%x_1x_2, x_1x_3, \ldots x_1x_n, x_1, x_2x_3, x_2x_4, \ldots
%\]
Consider the projection which omits the first coordinate:

\[
\pi_1 : \mathbb{F}^n_2 \rightarrow \mathbb{F}^{n-1}_2
\]
\[
(a_1 , a_2 ,\ldots, a_n )\mapsto (a_2,\ldots ,a_n),
\]
and denote by $B[2, n]$ the ring of Boolean polynomials where $x_1$ has been omitted. We may in a similar fashion consider a sequence of $k$ projections

\[
\mathbb{F}^n_2 \rightarrow \mathbb{F}^{n-1}_2\rightarrow\cdots\rightarrow \mathbb{F}^{n-k}_2,
\]
where the $i$'th projection is denoted $\pi_i:\mathbb{F}^{n-i+1}_2 \rightarrow \mathbb{F}^{n-i}_2$ for $1\leq i\leq k$. We denote the ring of Boolean functions where we omit the sequence of variables $x_1,\ldots ,x_k$ as $B[k+1,n]$.

\subsection{Systems of Boolean equations and ideals}

%Any Boolean function $f : \mathbb{F}^n_2 \rightarrow \mathbb{F}_2$ can be written as a polynomial in the ring $B[1,n]$. Such a function is uniquely determined by the subset
%$$Z = \{\mathbf{a} \in\mathbb{F}^n_2 | f (\mathbf{a}) = 0\},$$
%and they are therefore in one-one correspondence with subsets of $\mathbb{F}^n_2$. On the other hand, for $f\in B[1,n]$ we may write $f$ on the form $f=\sum\limits_{\alpha\in A} \prod\limits_{i=1}^{n} x_i^{\alpha_i}$, where $\alpha=\{\alpha_1,\ldots ,\alpha_n\}$ is a binary vector of length $n$ and $A\subset GF(2)^n$ specifies the monomials occurring in $f$.

%Any Boolean function can be written in this form, which gives a different one-to-one correspondence between Boolean functions and subsets of  $\mathbb{F}^n_2$.  We therefore have the following relationship 

%$$A \leftrightarrow f \leftrightarrow Z.$$

%We denote our set of polynomial equations $f_1=\ldots =f_m=0$ by $F=\{f_1,\ldots ,f_m\}$. 
The polynomials in $F=\{f_1,\ldots ,f_m\}$ generate an ideal $I= (f_1,\ldots , f_m) = I(F)$ in the ring $B[1,n]$.  Let $Z(I)$ denote the zero set of this ideal, i.e, the set of points

$$Z(I) = \{\textbf{a}\in \mathbb{F}^n_2 | f (\textbf{a}) = 0\text{ for every }f\in I\}.$$

\begin{Lemma}\label{lem:principal}
Let $f, g$ be Boolean functions in $B[1,n]$. Then the following ideals are equal:
$$(f, g) = (f g + f + g).$$
\end{Lemma}

\begin{proof}
Clearly $(f,g)\supseteq (fg + f + g)$.  Note that also $Z(f, g) = Z(fg + f + g)$, since it is easy to check that $f(a) = g(a) = 0$ if and only if $f (a)g(a) + f(a) + g(a) = 0$. Thus the zero set $Z(f)\supseteq Z(fg + f + g)$.  This in turn means that the Boolean function $f$ is a multiple $h(fg + f + g)$ for some other Boolean function $h$, and similarly for $g$. Thus $$(f,g)\supseteq (fg + f + g) \supseteq (f,g),$$ which shows that these ideals are equal.
\end{proof}

\begin{Corollary}\label{cor:principal}
Any ideal $I = (f_1,\ldots,f_m)$ in $B[1,n]$ is a principal ideal. More precisely $I =(f)$ where
$$f=1+\prod_{i=1}^m(f_i+1).$$
\end{Corollary}

\begin{proof}
Let $I = (f_1 ,\ldots, f_m)$. By Lemma \ref{lem:principal} this is equal to the ideal $(f_1f_2 + f_1 + f_2 , f_3 ,\ldots,f_m)$, with one generator less. We may continue the process for the remaining generators providing us in the end with $I = (f)$, where

$$f=1+\prod_{i=1}^m(f_i + 1).$$

\end{proof}

\begin{Corollary}\label{cor2}
For two ideals in $B[1,n]$ we have $I\supseteq J$ if and only if $Z(I)\subseteq Z(J)$.
In particular $I = J$ if and only if $Z(I) = Z(J)$.
\end{Corollary}

\begin{proof}
By Corollary \ref{cor:principal} we have $I = (f)$ and $J = (g)$, where $f$ and $g$ are the respective principal generators. Clearly if $(f)\supseteq (g)$ then $Z(g)$ contains $Z(f)$. If the zero set of $g$ contains the zero set of $f$, then $g$ = $fh$ for some polynomial $h$. Hence $(f)\supseteq (g)$.
\end{proof}

\medskip
Now given an ideal $I\subset B[1,n]$, our aim is to find the ideal $I_2\subset B[2, n]$ such that $Z(I_2)=\pi_1(Z(I))$. More generally, when eliminating more variables we aim to find the ideal $I_{k+1}\subset B[k+1,n]$, such that $Z(I_{k+1})=\pi_{k}\circ(\cdots\circ(\pi_1(Z(I))))$. Since the complexity of the polynomials can grow very quickly when eliminating variables, we would rather want to compute an ideal $J$, as large as possible given computational restrictions, which is contained in $I_{k+1}$.% ($I_2$ resp).

If we can find solutions to $Z(J)$ contained in $\mathbb{F}^{n-k}_2$ we can then check if they lift to solutions in $Z(I)$ contained in $\mathbb{F}^{n}_2$, by sequentially lifting the solutions backwards with respect to each projection. Let us first describe precisely the ideal $I_{k+1}$ whose zero set is the sequence of projections $\pi_{k}\circ(\cdots\circ(\pi_1(Z(I))))$. This corresponds to what is known as the \emph{elimination ideal} $I\cap B[k+1,n]$.

\begin{Lemma}\label{Lemma1}
Let $I_{k+1}\subseteq B[k+1,n]$ be the ideal of all Boolean functions vanishing on $\pi_{k}\circ(\cdots\circ(\pi_1(Z(I))))$. Then $I_{k+1} = I\cap B[k+1,n]$.
\end{Lemma}

\begin{proof}
We show this for the case when eliminating one variable, the general case follows in a similar manner. Clearly $I_2\supseteq I\cap B[2, n]$. Conversely let $f\in B[2, n]$ vanish on $\pi_1 (Z(I))$. Then $f$ must also vanish on $Z(I)$, where $f$ is regarded as a member of the extended ring $B[1,n]$. Therefore $f\in I$ by Corollary \ref{cor2}.
\end{proof}

A standard technique for computing elimination ideals is to use Gr\"{o}bner bases, which eliminate one \emph{monomial} at the time. %In fact, to compute elimination ideals via Gr\"{o}bner bases one has to compute the full Gr\"{o}bner basis before performing elimination. Also, Gr\"{o}bner bases are the only known tool which can eliminate from more general systems of polynomials. (I.e overdetermined systems etc) 
Computing Gr\"{o}bner bases is computationally heavy because the degrees of the polynomials grow rapidly over the iterations. To deal with this problem we propose two algorithms which attempt to limit the degrees of polynomials that arise. %tries to handle the growing degrees during elimination. 
Our solution is to not use all elements during elimination, but {\it discard high degree polynomials and only keep the low-degree ones.} We denote an ideal where the degree is restricted to some $\delta$ by $J^{\delta}$, whereas $J^{\infty}$ means that we allow \emph{all} degrees.

The benefit from our solution is that the elimination process gives us an algorithm with much lower complexity, at the cost of the following two disadvantages:

\begin{enumerate}
    \item %Firstly, t
    Discarding polynomials of degree $ > \delta$ %The theory of elimination tells us that straight forward elimination where we omit higher degree terms 
    gives an ideal $J^{\delta}$ that is only contained in the elimination ideal $J^{\infty} = I\cap B[j+1, n]$ for $1< j\leq k$. It follows that $Z(J^{\delta})$ of the eliminated system contains all the projected solutions of the original set of equations, but it will also contain “false” solutions which will not fit the ideal $I$ when lifted back to ${\mathbb F}^n_2$, regardless of which values we assign to the eliminated variables.
    \item %Secondly, 
    Since the proposed algorithms expand the solution space to include false solutions, the worst case scenario is when we end up with an empty set of polynomials after eliminating a sequence of variables. This means that all constraints given by the initial $I$ have been removed, and we end up with the complete $\mathbb{F}_2^{n-k}$ as a solution space.
    %\item It is important to note that not discarding any polynomials will provide only the true solutions to the set where variables have been eliminated, which then can be lifted back to the solutions of the initial ideal $I(F)$. The drawback of this approach is that we must be able to handle arbitrarily large polynomials, i.e high computational complexity.
\end{enumerate}

It is important to note that \emph{not discarding} any polynomials will provide only the true solutions to the set where variables have been eliminated, which then can be lifted back to the solutions of the initial ideal $I$. The drawback of this approach is that we must be able to handle arbitrarily large polynomials, i.e high computational complexity.

Thus there is a tradeoff between the maximum degree $\delta$ allowed, and the proximity between the ``practical'' ideal $J^{\delta}$ %Moreover, this brings us into a trade-off situation between the ability to control the degree and the ability to stay close to 
and the true elimination ideal $I\cap B[k+1,n]$. %In this paper the goal is to optimize the trade-off between computational complexity and the degree of the output polynomials. With respect to complexity we do this by only considering polynomials of degree $3$ and less. Computations that would lead to higher degree polynomials are subject to future research. In this regard we propose two different algorithms, one which is optimal with respect to degree and another which is a more refined and efficient algorithm. The highlight is where we show that the proposed algorithms gives not only a complete mathematical framework, but also a deeper understanding of the concept of elimination of variables from systems of Boolean functions.

%For polynomials of degree at most $\delta = 3$ we introduce the following notation. Any polynomial in $B[1,n]$ of degree $3$ contain at most  $\sum_{i=0}^3\binom{n}{i}$  monomials.  For convenience we define  a sum of binomial coefficients as

%$$\binom{n}{\leq j}=\sum_{i=0}^j\binom{n}{i}.$$
In this paper we limit the degree to $\delta=3$, and we therefore consider the two sets
\[
F^3 = \{f^3_1 ,\ldots , f^3_{r_3}\}, F^2 = \{f^2_1,\ldots , f^2_{r_2}\}
\]
of polynomials, where the $f_i^3$'s all have degree $3$ and the $f_i^2$'s have degree $\leq 2$. Furthermore, the polynomials in $F^3$ and $F^2$ together generate the ideal $I = (F^3,F^2)$. We use the notation $F^i_{x_1},F^i_{\overline{x_1}}$, ($i=2,3$) to distinguish disjoint subsets that contain (resp. do not contain) any monomial with the variable $x_1$.  
In the following sections we are going to develop the mathematical framework for computing ideals $I_2,\ldots,I_k$ such that $I_j\subseteq I\cap B[j+1,n]$ starting from the ideal $I$, and such that the generators for each $I_j$ has degree $\leq 3$.

\section{Elimination Techniques}
\label{sec:elimination}
Several methods for solving systems of Boolean equations based on various approaches have been suggested. In \cite{XL} and \cite{XSL} the authors introduce the XL and XSL algorithms, respectively.  The basic idea is to multiply equations with enough monomials to re-linearize the whole system of Boolean equations. Our approach can be viewed as a specialized generalization of this approach.  By this we mean that since we bound the degree at $\leq 3$, the quadratic equations will only be multiplied with linear monomials.  The aim is to squeeze out as many quadratic and cubic polynomials which in turn can be used to solve the system of Boolean equations. Furthermore, we note that the elimination aspect is not considered in the XL and XSL approaches.

Our objective is to find as many polynomials in the ideal $I$ generated by $F^2$ and $F^3$ %of degree $\leq 3$ 
as possible, \emph{computing only with polynomials of degrees $\leq 3$}. This limits both the storage and computational complexity. % at a minimum. 
Our approach to solve the system of equations 
\[ f_i^{\delta} =  0 , \delta=2,3\mbox{ and }i=1, \ldots, r_{\delta}\]
is to eliminate variables so that we find degree $\leq 3$ polynomials in $I_k$, in smaller and smaller Boolean rings $B[k+1, n]$. We introduce the algorithms for doing this in Section \ref{Algorithms}.

Let $F=(f_1,\ldots,f_m)$ be a set of Boolean functions in $B[1,n]$ of degree $\leq 3$, and denote by $\langle F\rangle$ the vector space spanned by the polynomials in $F$, where each monomial is regarded as a coordinate. Let $L = \{1, x_1 ,\ldots , x_n \}$ consist of the constant and linear functions in $B[1,n]$, such that $\langle L\rangle$ is the vector space spanned by the Boolean polynomials of degree $\leq 1$. For the set of quadratic polynomials $F^2$, we denote the product $LF^2$ as the set of all products $lg$ where $l\in L$ and $g\in F^2$. Then it suffices to eliminate variables from the vector space $\langle F^3\cup LF^2\rangle$. For convenience we let the set $F=(f_1,\ldots,f_m)$ be the set of cubic polynomials generated by $F^3\cup LF^2$.

Hence, as part of the variable elimination process, it will be necessary to split a set of polynomials according to different criteria. Two procedures, Algorithm~\ref{alg:splitx1}:$SplitVariable()$ and  Algorithm~\ref{alg:split23}:$SplitDeg2/3()$ are described in Appendix~A. $SplitVariable(F,x_1)$ splits a set $F$ of polynomials into the subsets $F_{x_1},F_{\overline{x_1}}$, where $F_{x_1}$ only has polynomials that contain $x_1$ and $F_{\overline{x_1}}$ consists of all polynomials not containing $x_1$. $SplitDeg2/3(F)$ splits a set $F$ of polynomials into the subsets $F^2, F^3$, where polynomials in $F^2$ are quadratic, and polynomials in $F^3$ have degree $3$.

These two procedures can essentially be implemented in terms of row reduction on the incidence matrices of $F$, where  the monomials (i. e. columns of the matrices) are ordered lexicographically and by degree, respectively. More details are provided in Appendix~A.

%\subsection{Resultants and coefficient constraints} \label{subsec:RC}

In order to eliminate variables from a system of Boolean functions we are going to use resultants, which eliminate one variable at the time from a pair of equations. Let $f_1=a_1x_1+b_1$ and $f_2=a_2x_1+b_2$ be two polynomials in $B[1,n]$, where the variable $x_1$ has been factored out. Then the polynomials $a_j$ and $b_j$ are in $B[2, n]$. In order to find the resultant, form the $2\times 2$ Sylvester matrix of $f_1$ and $f_2$ with respect to $x_1$
\[
\syl(f_1,f_2,x_1) = \left(\begin{array}{ccccc}
a_1 & a_2 \\
b_1 & b_2 \\
\end{array}\right)
\]

The \emph{resultant} of $f_1$ and $f_2$ with respect to $x_1$ is then simply the determinant of this matrix, and hence a polynomial in $B[2,n]$:
\[
\Res (f_1,f_2,x_1)= \det (\syl(f_1,f_2,x_1)) = a_1b_2+a_2b_1
\]

We note that $\Res(f_1, f_2, x_1 ) = a_2f_1 + a_1f_2$, which means that the resultant is indeed in the ideal generated by $f_1$ and $f_2$. Moreover, $\Res(f_1, f_2, x_1)$ is in the elimination ideal $I_2$ and when both $f_i$'s are quadratic, then the $a_i$'s are linear so the degree of the resultant $\Res(f_1,f_2,x_1)$ is $\leq 3$.

This observation gives hope for computational purposes, since $2\times 2$ determinants are easy to compute and cubic polynomials can be handled by a computer, also for the number of variables $n$ we encounter in cryptanalysis of block ciphers.  Note that one resultant computation eliminates all monomials containing the targeted variable, and not only the leading term, as in Gr\"{o}bner basis computations.

For an ideal $I(F) = (f_1 ,\ldots, f_m )\subseteq B[1,n]$ generated by a set $F$ of Boolean polynomials, we can compute the resultant of every pair of polynomials $\Res(f_i , f_j ; x_1 )$, which in turn gives us the ideal of resultants:

\[
\Res_2(F;x_1)=(\Res(f_i,f_j;x_1)|1\leq i<j\leq m).
\]

It is easy to show that the ideal of resultants is contained in the elimination ideal $I(F)\cap B[2, n]$, but this inclusion is in general strict. To close the gap we need the following ideal:

\begin{Definition}
Let $I(F) = (f_1 ,\ldots , f_m)\subseteq B[1,n]$, and write each $f_i$ as $f_i = a_ix_1 + b_i$, where $x_1$ does not occur in $a_i$ or $b_i$. We define the \emph{coefficient constraint} ideal:

\[
\Co_2(F)=(b_1 (a_1 + 1), b_2 (a_2 + 1),\ldots , b_m (a_m + 1)).
\]

\end{Definition}

Note that the degrees of the generators of $\Co_2(F)$ have the same degrees as the generators of the resultant ideal.  In the case when $I(F)$ consists of quadratic polynomials, the generators of $\Co_2(F)$ will be polynomials of degree $\leq 3$. %One crucial fact is that 
The zero set of this ideal lies in the projection of the zero set of $I(F)$ onto $\mathbb{F}_2^{n-1}$.

\begin{Lemma}\label{Lemma2}
$Z(\Co_2(F))\supseteq \pi_1(Z(I(F))).$
\end{Lemma}

\begin{proof}
Note that a point $\textbf{p}\in \mathbb{F}^{n-1}_2$ is not in the zero set $Z(\Co_2(F))$ only  if for some $i$ we have $a_i(\textbf{p}) = 0$ and $b_i (\textbf{p}) = 1$. But then for both the two liftings of $\textbf{p}$ to $\mathbb{F}^n_2$: $\textbf{p}_0=(0,\textbf{p})$ and $\textbf{p}_1=(1,\textbf{p})$ we have $f_i(\textbf{p}_j) = 1$. Therefore $\textbf{p}\notin \pi_1 (Z(I(F)))$, and so we must have $Z(\Co_2(F))\supseteq \pi_1(Z(I(F)))$.
\end{proof}

By Lemmas \ref{Lemma1} and \ref{Lemma2}, the coefficient constraint ideal is in the elimination ideal. We can now use this ideal to describe the full elimination ideal, which turns out to be generated exactly by $\Res_2(F)$ and $\Co_2(F)$.

%Let $I=\langle f_1,\ldots ,f_s \rangle\subset S$ be an ideal, and let $f_i=a_ix_k+b_i$ be the factorization of the polynomial $f_i$'s with respect to the variable $x_k$.  Then $I(J)_{x_k}+\resideal_{x_k}= I\cap S_{x_k}$.

%We show the equality of  and $I\cap S_{x_k}$ by looking at the zero-sets of polynomials.  For $f\in GF(2)[x_1,\ldots,x_n]$ the set of points satisfying $f=0$ is denoted by $V(f)$.
%It is now easy to prove that the elimination ideal after eliminating $x_k$ is equal to the ideal generated by both the resultants and the $J$-ideal for $x_k$.

\begin{Theorem}\label{prop:elim}

Let $I(F)=(f_1 ,\ldots , f_m)\subseteq B[1,n]$ be an ideal generated by a set $F$ of Boolean polynomials. Then 
\[
I(F )\cap B[2, n] = I(\Res_2(F ), Co_2(F )).
\]
\end{Theorem}

\begin{proof}
By Lemma \ref{Lemma1} we have

\[
\pi_1(Z(I(F)))=Z(I(F)\cap B[2,n]).
\]

We know that

\[
I(F)\cap B[2,n]\supseteq I(\Res_2(F),\Co_2(F)),
\]

which implies that

\[
\pi_1(Z(I(F)))=Z(I(F)\cap B[2,n])\subseteq Z(\Res_2(F))\cap Z(\Co_2(F)).
\]

Conversely, let a point $\textbf{p}\in \mathbb{F}^{n-1}_2$ in the right hand side above be given. Then it has two liftings to points in $\mathbb{F}^n_2$: $\textbf{p}_0=(0,\textbf{p})$ and $\textbf{p}_1=(1,\textbf{p})$. Let $f_i=x_1a_i+b_i$ be an element in $F$. Since $\textbf{p}$ vanishes on $\Co_2(F)$, the following are the possible values for the terms in $f_i$ when applied to the lifting $\textbf{p}_j$.

$$\begin{array}{ccc}
a_i&b_i&x_1\\
\hline
0&0&0\\
0&0&1\\
1&0&0\\
1&1&1\\ 
\end{array}$$

Note that $\Co_2(F)$ excludes $a_i(\textbf{p})$ and $b_i(\textbf{p})$ from taking the values $0,1$.  Since $\textbf{p}$ vanishes on the resultant ideal, there cannot be two $f_i$ and $f_j$ such that $a_i(\textbf{p}),b_i(\textbf{p})$ takes values $1,0$ and $a_j(\textbf{p}),b_j(\textbf{p})$ takes values $1,1$, since in that case the resultant $a_ib_j+a_jb_i$ does not vanish. This means that the values of $a_i(\textbf{p}),b_i(\textbf{p})$ are either $i)$ All $0,0$ or $1,0$, or $ii)$ All $0,0$ or $1,1$. In case $i)$, the lifting $\textbf{p}_0$ is in the zero set $Z(I(F))$. In case $ii)$ the lifting $\textbf{p}_1$ is in the zero set of $Z(I(F))$. This shows that $Z(\Res_2(F))\cap Z(\Co_2(F))$ lifts to $Z(I(F))$, which means that

\[
\pi_1(Z(I(F)))=Z(I(F)\cap B[2,n])\supseteq Z(\Res_2(F))\cap Z(\Co_2(F))
\]

as desired.
\end{proof}

Note that since we limit the degree to $\leq 3$, we only compute resultants and coefficient constraints with respect to the set $F^2$ consisting of quadratic Boolean functions. The process of producing the resultants and the polynomials of the coefficient constraints of the set $F^2$ with respect to the variable $x_1$ is described in Algorithm \ref{alg:res}.

\begin{algorithm}
   \caption{$RandC(F_{x_1}^2,x_1)$}
   \label{alg:res}
    \begin{algorithmic}
      \Require{$F_{x_1}^2=(f^2_1,\ldots,f^2_{r_2})$ set of quadratic polynomials in $B[1,n]$}
      \Ensure{Set $R$ of cubic polynomials where $\pi_1(Z(F_{x_1}^2))=Z(R)$ and $x_1\not\in R$}
      \State 
      \State $f^2_i=a_ix_1+b_i$ for $1\leq i\leq r_2$
      \State $R=\emptyset$
      \For{$(f^2_i,f^2_j)\in F_{x_1}^2\times F_{x_1}^2, f^2_i\neq f^2_j$}
        \State $R\leftarrow R\cup \{a_ib_j+a_jb_i$\}
    \EndFor
    \For{$f^2_i\in F_{x_1}^2$}
        \State $R\leftarrow R\cup \{b_i(a_i+1)\}$
    \EndFor
    \State Return $R$
    \end{algorithmic}
\end{algorithm}

In general, for an ideal $I(F)=(f_1 ,\ldots , f_m )\subseteq B[1,n]$, this process can obviously be iterated eliminating more variables from $I(F)$. We denote by the ideals $\Res_{k+1}(F)$ and $Co_{k+1}(F )$ the iterative application of the resultant and the coefficient constraint ideal with respect to a sequence $x_1,\ldots ,x_k$ of variables to be eliminated, with the initial polynomials from $I(F)$ as input. Note that both Lemma \ref{Lemma2} and Proposition \ref{prop:elim} easily generalize to this case. Hence we generalize Theorem \ref{prop:elim} as follows.

\begin{Corollary}
\label{cor:elim}
For $I(F) = (f_1 ,\ldots , f_m)$ in $B[1,n]$, then
\[
I(F )\cap B[k+1, n] = I(\Res_{k+1}(F ),Co_{k+1}(F )).
\]
\end{Corollary}

It is important to note that Corollary \ref{cor:elim} is only valid if we allow the degrees to grow with each elimination, but including the coefficient constraints solve the problem with the resultant only being a subset of the elimination ideal. This enables us to actually compute the elimination ideal \emph{not depending on any monomial order}. Moreover, one could find the elimination ideal by successively eliminating $x_1,\ldots ,x_k$ using Corollary \ref{cor:elim} by the following algorithm:
\begin{enumerate}
   \item $F_1=F$,
   \item $F_2=$ generators of $\Res_2(F_1) + Co_2(F_1)$,
   \item $F_3=$ generators of $\Res_3(F_2) + Co_3(F_2)$,
   \item $\cdots$
\end{enumerate}

However, applying this strategy in practice leads to problems due to the growing degrees of the resultants and coefficient constraints with each elimination. In fact if the $f_i$'s have degree $d$, the degrees of the resultants and the coefficient constraints have degree $2d-1$. With many variables the size of the polynomials and the number of monomials quickly become too large for a computer to work with. This is the reason why we limit the degree at $\leq 3$, enabling us to deal with these complexity issues. %and optimizing the computational complexity vs the degree trade-off of the output polynomials.

\section{Elimination Algorithms}\label{Algorithms}

\subsection{The L-Elim{\bf A}-algorithm}

In the following we are going to apply the procedure \noindent {\bf A.} below as a building block in order to produce as many polynomials of degree $\leq 3$ as we can when eliminating variables from the sets $F^3$ and $F^2$.

\medskip

\noindent {\bf A.} %For the set $F^2$ we can compute the intersection $\langle F^2_{\overline{x_1}}\rangle = \langle F^2 \rangle \cap \Bton$ by doing Gaussian elimination on the polynomials in $F^2$ as described in $SplitVariable()$ in Appendix~A. These polynomials do not contain $x_1$, and $\langle F^2\rangle=\langle F^2_{x_1}\cup F^2_{\overline{x_1}}\rangle$. Using this we can now 
We compute two sets $F^2_{\overline{x_1}}$ and $F^3_{\overline{x_1}}$ of quadratic and cubic polynomials in $\Bton$, that satisfy

\[\langle F^2_{\overline{x_1}} \cup F^3_{\overline{x_1}}\rangle =\langle F^3 \cup L F^2 \rangle \cap \Bton. \] 

This gives a new pair of sets $F^3_{\overline{x_1}}$ and $F^2_{\overline{x_1}}$, but now in the smaller ring $\Bton$. This procedure can be continued giving $F^3_{\overline{x_1},\overline{x_2},\ldots,\overline{x_k}}, F^2_{\overline{x_1},\overline{x_2},\ldots,\overline{x_k}}$ in smaller and smaller Boolean rings $B[k+1,n]$. The computation of $F^2_{\overline{x_1}}$ and $F^3_{\overline{x_1}}$ is the main objective of our algorithm. In Algorithm \ref{alg:elimv1x1} we show how we perform this procedure in L-Elim{\bf A}.

\begin{algorithm}
    \caption{$L-Elim{\bf A}(F^3,F^2,x_1)$}
    \label{alg:elimv1x1}
    \begin{algorithmic}
    \Require{$F^3=(f^3_1,\ldots,f^3_{r_3})$ set of cubic polynomials in $B[1,n]$, $F^2=(f^2_1,\ldots,f^2_{r_2})$ set of quadratic polynomials in $B[1,n]$, and $x_1$ the variable to be eliminated from $F^3$ and $F^2$}
    \Ensure{Set $F^3_{\overline{x_1}}$ of cubic polynomials and set $F^2_{\overline{x_1}}$ of quadratic polynomials, where $x_1\not\in F^2_{\overline{x_1}} \cup F^3_{\overline{x_1}}$}
        \State $L=\{1, x_1 ,\ldots , x_n \}$ 
        \State $F^*\leftarrow F^3\cup L\cdot F^{2}$ 
        \State $F^{2}, F^3 \leftarrow SplitDeg2/3(F^*)$
        
        \State $F^2_{x_1}, F^2_{\overline{x_1}} \leftarrow SplitVariable (F^2,x_1)$
        \State $F^3_{x_1}, F^3_{\overline{x_1}} \leftarrow SplitVariable (F^3,x_1)$ 
       \State Return $F^3_{\overline{x_1}}, F^2_{\overline{x_1}}$
    \end{algorithmic}
\end{algorithm}

\begin{Remark}{\bf Computational complexity:}
The heaviest step in $L-Elim{\bf A}$ is the call to $SplitDeg2/3()$.  This procedure does Gauss elimination on $\mathcal{O}(n^3)$ columns in a matrix with $\mathcal{O}(n^3)$ rows.  In total we need $\mathcal{O}(n^9)$ bit operations for running $L-Elim{\bf A}.$
\end{Remark}

\begin{Remark}\label{REMARK1}
Algorithm \ref{alg:elimv1x1} is related to the XL and XLS algorithm \cite{XL},\cite{XSL}, where we restrict the algorithm here by only multiplying with linear Boolean monomials $L$. However, our approach contains the aspect of elimination of variables which is not considered in XL-type algorithms. 
\end{Remark}

Next, we proceed to develop a more efficient algorithm. In fact, the next construction optimizes the elimination approach by showing that we do not need to multiply with all variables as done in $L-Elim{\bf A}()$.

\subsection{Main elimination algorithm {\bf A}}

Motivated by Algorithm~\ref{alg:elimv1x1}, we aim in a similar fashion to produce more polynomials of degree $\leq 3$ for the sets $F^3,F^2$, but in a more efficient way. We use resultants, coefficient constraints and some other steps to be detailed below. The ensuing algorithm $eliminate{\bf A}()$ is presented in Algorithm~\ref{alg:main}.

The inputs to $eliminate{\bf A}()$ are sets $F^3$ and $F^2$, of polynomials of degree 3 and 2 respectively.  These sets will be modified in the algorithm to only include polynomials without $x_1$, the variable to be eliminated.  We also want the sets to be as large as possible, but of course only with linearly independent polynomials.

$eliminate{\bf A}()$ starts by splitting $F^2$ into subsets $F^2_{x_1}$ containing $x_1$ and $F^2_{\overline{x_1}}$ not containing $x_1$, using $SplitVariable()$.  These sets are first used to increase $F^3$, by adding $x_1F^2_{\overline{x_1}}$ and $(x_1+1)F^2_{x_1}$ to $F^3$.  Note that we multiply $F^2$ with only one variable ($x_1$, to be eliminated), and not all of $L$ as in $L-Elim{\bf A}()$.  This leads to much lower space complexity in $eliminate{\bf A}()$, since the $F^3_{\overline{x_1}}$ computed here is much smaller than the one output from $L-Elim{\bf A}()$. Then we split $F^3$ into subsets $F^3_{x_1}$ containing $x_1$ and $F^3_{\overline{x_1}}$ not containing $x_1$, using $SplitVariable()$.

Next, $F^3_{x_1}$ is normalized (see Algorithm~\ref{alg:norm}: $Normalize()$ in Appendix~B) with $F^2_{x_1}$ as basis.  This removes many of the degree $3$ monomials containing $x_1$ in $F^3_{x_1}$, and if a polynomial on the special form $f_{\infty}=x_1+g(x_2,\ldots,x_n)$ is found in $F^2_{x_1}$ all degree $2$ monomials $x_1x_i$ will be eliminated from $F^3_{x_1}$. The output of $Normalize()$ is $F^{3,norm}_{x_1}$ and $F^{3,norm}_{\overline{x_1}}$, and we join $F^{3,norm}_{\overline{x_1}}$ to $F^3_{\overline{x_1}}$.

Now we compute resultants and coefficient constraints from the set $F^2_{x_1}$, creating a set $R$ of cubic polynomials in $\Bton$ by using Algorithm~\ref{alg:res} $RandC()$. This produces as many polynomials without $x_1$ as possible. All of these are added to $F^3_{\overline{x_1}}$ and we remove potentially linearly dependent polynomials in $F^3_{\overline{x_1}}$ such that

\begin{equation}\label{newF2F3}
F^3_{\overline{x_1}}:=F^3_{\overline{x_1}}\cup F^{3,norm}_{\overline{x_1}}\cup R
\end{equation}

The two new sets $F^2_{\overline{x_1}},F^3_{\overline{x_1}}$ in $\Bton$, neither containing $x_1$, are returned from $eliminate{\bf A}()$, as described in Algorithm \ref{alg:main}.
In the following we show that in fact the output of $eliminate{\bf A}()$ and $L-Elim{\bf A}()$ is the same.

\begin{Theorem} \label{pro:FLG} 
Let $F^2, F^3$ be the input to $eliminate{\bf A}()$.  
Let $F^2_{\overline{x_1}}$ be the subset of $F^2$ not containing $x_1$ and let $F^3_{\overline{x_1}}$ be the defined as in (\ref{newF2F3}). Then
\[ \langle F^3_{\overline{x_1}} \cup L_{\overline{x_1}} F^2_{\overline{x_1}} \rangle = 
\langle F^3 \cup L F^2 \rangle \cap \Bton.
\]
\end{Theorem}

\begin{proof}
The fact that $\langle F^3 \cup L F^2 \rangle \cap \Bton \supseteq \langle F^3_{\overline{x_1}} \cup L_{\overline{x_1}} F^2_{\overline{x_1}} \rangle$ is obvious from the construction.

To prove the converse, if a polynomial $f^2 \in F^2_{x_1}$ has leading term $x_1x_i$ we denote the polynomial as $f^2 = f^2_i=a_ix_1+b_i$. Similarly, if it has leading term $x_1$  we denote it by $f^2 = f^2_\infty$. With this notation we let  $I \sus \{ 2, \ldots,n\} \cup \{\infty \}$ be the index set of these polynomials, such that $F^2_{x_1} = \{ f^2_i \}_{i \in I}$. Let $F^{2}_{\overline{x_1}}$ be $\{ f^{2}_j\}_{j \in J}$ for some index set $J$ such that $J\cap I=\emptyset$. 

Let $p$ be a polynomial in $\langle F^3 \cup L F^2 \rangle \cap B[2,n]$. We can then write  
\[
p = f^3 + f^{3}_{x_1} + f^{3}_{\overline{x_1}}
\]
where $f^3 \in F^3$, $f^3_{x_1} \in \langle L \cdot  F^2_{x_1} \rangle$, and $f^{3}_{\overline{x_1}} \in \langle L \cdot  F^{2}_{\overline{x_1}} \rangle$. The goal is to subtract from $p$ the terms in $F^3_{\overline{x_1}} \cup L_{\overline{x_1}} F^2_{\overline{x_1}}$ which is produced according to (\ref{newF2F3}). In the end we will show that we are left with $p=0$, which proves that $p$ is originally in $\langle F^3_{\overline{x_1}} \cup L_{\overline{x_1}} F^2_{\overline{x_1}}\rangle$.

With $p$ as above we can find index sets $I^0, I^1 \sus I$ and $J^0 \sus J$, and
$K \sus \{2,\ldots, n\}  \times I$ such that 

\begin{align*}
f^3_{x_1} & = \sum_{i \in I^0} (x_1 +1)f^2_i + \sum_{i \in I^1} f^2_i + \sum_{(k,i) \in K} x_k f^2_i, \\
f^{3}_{\overline{x_1}} & = \sum_{j \in J^0} x_1 f^{2}_j + f^{3\prime},
\end{align*}

where $f^{3\prime} \in \langle \{ 1, x_2, \ldots, x_n \} F^2_{\overline{x_1}} \rangle = \langle L_{\overline{x_1}} \cdot F^2_{\overline{x_1}} \rangle$.

By the normalization we may write the following part of $p$:
\[
f^3 + \sum_{i \in I^0} (x_1 +1)f^2_i 
+ \sum_{j \in J^0} x_1 f^{2}_j
\]
as 

\[
f^{3,norm}_{x_1} + \sum_{(k,i) \in K^\prime} x_kf^2_i + f^{3,norm}_{\overline{x_1}} 
\]
for some index set $K^\prime$ with all $k \in \{2, \ldots, n\}$, and where $f^{3,norm}_{x_1} \in \langle F^{3,norm}_{x_1} \rangle$ and $f^{3,norm}_{\overline{x_1}} \in \langle F^3_{\overline{x_1}} \cup F^{3,norm}_{\overline{x_1}} \rangle$. So we have 
\begin{equation} \label{eq:Algfgg}
p  = f^{3,norm}_{x_1} + \sum_{i \in I^1} f^2_i + \sum_{(k,i) \in K^{\prime\prime}} x_j f^2_i + f^{3\prime} + f^{3,norm}_{\overline{x_1}} 
\end{equation}
where $K^{\prime \prime}=K\cup K^{\prime}$. 

Now we consider the terms above with $f^2_2$. Suppose $x_2 f^2_2$ occurs in the sum on the right. Note that $x_1 x_2$ is not in $p$ (since $p \in B[2,n]$), and not in $f^{3,norm}_{x_1}$ (since it is $3$-normal). So $x_1x_2$ in $x_2 f^2_2$ must cancel against $x_1 x_2$ in $f^2_2$. Then $f^2_2 + x_2f^2_2 = (x_2 +1) f^2_2$ occurs. Now
\[
(x_2+1)f^2_2 = (a_2 +1) f^2_2 + \mbox{smaller terms than } x_1x_2.
\]
We now subtract $(a_2+1)f^2_2=(a_2+1)b_2$ (which is in $\Bton$) from both sides of 
\eqref{eq:Algfgg}. So its new left side is:
\[
p := p - (a_2+1)f^2_2, 
\] 
and the new right side of \eqref{eq:Algfgg} will contain neither $f^2_2$ nor $x_2 f^2_2$
(while it may contain $x_j f^2_2$ for $j \geq 3$).
It follows that the left side of \eqref{eq:Algfgg} above will not contain $x_2 f^2_2$. 

If now in the new equation \eqref{eq:Algfgg}, $x_2 f^2_3$ occurs, then $x_1x_2x_3$ must cancel against $x_3f^2_2$ (since $x_2 f^2_2$ does not occur any more on the right side of \eqref{eq:Algfgg}, and $x_1x_2x_3$ does not occur in $p$ nor in $f_{\overline{x_1}}^{3,norm}$). We substract the resultant $a_2 f^2_3 + a_3f^2_2$ from both sides
of \eqref{eq:Algfgg}. So its new left side is:
\[
p := p - (a_2 f^2_3 + a_3 f^2_2).
\] 

Then neither $x_2 f^2_2$ nor $x_2 f^2_3$ will occur anymore on the right side of \eqref{eq:Algfgg}. In this way we continue and no term $x_2 f^2_i$ with $2 \leq i \leq n$ will occur in the right side of \eqref{eq:Algfgg}. If $x_2 f^2_\infty$ occurs, then $x_1 x_2$ must cancel against the same term in $f^2_2$. We then subtract the corresponding resultant from both sides of \eqref{eq:Algfgg} and its new left side is:
\[
p := p - (a_2 f^2_{\infty} + f^2_2).
\]
We may then assume that no terms in the right side of \eqref{eq:Algfgg} contains any $x_2 f^2_i$. Then note the following: The right side of \eqref{eq:Algfgg} does not contain $x_3 f^2_2$ since if it did the term $x_3 x_1 x_2$ could not cancel against anything else on the right side.

\medskip
Now we continue and remove terms $x_3 f^2_i$ from the right side of \eqref{eq:Algfgg}. Then we remove terms $x_4 f^2_i$ and so on. Considering the $p$ in \eqref{eq:Algfgg}, we then get that modulo $\Co_{2}$ and 
$\Res_{2}$ we can write it as

\[
p = f^{3,norm}_{x_1} + \sum_{j \in I^\prime} f^2_j + f^{3\prime} + f^{3,norm}_{\overline{x_1}} 
\]
for some $I^\prime \sus I$. If $f^2_2$ occurs, the terms $x_1x_2$ could not cancel against anything in $f^{3,norm}_{x_1}$, thus $f^2_2$ does not occur. If $f^2_3$ occurs, the term $x_1 x_3$ could not cancel against anyting on the right side above, since $f^2_2$ does not occur. Hence $f^2_3$ does not occur. In this way we could continue and in the end get 

\[
p = f^{3,norm}_{x_1} + f^{3\prime} + f^{3,norm}_{\overline{x_1}}.
\]
But then clearly $f^{3,norm}_{x_1}$ equals $0$. Thus modulo 
$\langle \Co_{2} \cup \Res_{2} \rangle$ the original $p$ is in $\langle L_{\overline{x_1}} \cdot F^2_{\overline{x_1}} \rangle + \langle F^3_{\overline{x_1}} \cup F^{3,norm}_{\overline{x_1}} \rangle$. This proves the Theorem.
\end{proof}

\begin{Corollary}\label{pro:FLG:Severalvariables} 
Let $L_{\overline{x_1},\ldots ,\overline{x_{k}}} = \{1, x_{k+1}, \ldots x_n \}$ and $F^3_{\overline{x_1},\ldots ,\overline{x_{k}}}$, $F^2_{\overline{x_1},\ldots,\overline{x_k}}$ be the result of applying $eliminate{\bf A}()$ $k$ times to the input sets $F^3$, $F^2$. Then
\[
\langle F^3 \cup L F^2 \rangle \cap B[k+1,n] = \langle F^3_{\overline{x_1},\ldots ,\overline{x_{k}}} \cup L_{\overline{x_1},\ldots ,\overline{x_{k}}} F^2_{\overline{x_1},\ldots ,\overline{x_{k}}} \rangle.
\]

\end{Corollary}

\begin{proof}
Given a sequence of variables $x_1,x_2,\ldots ,x_k$ to be eliminated. Then applying Theorem \ref{pro:FLG} on $x_1$ gives us $\langle F^3 \cup L F^2 \rangle \cap \Bton=\langle F^3_{\overline{x_1}} \cup L_{\overline{x_1}} F^2_{\overline{x_1}} \rangle$. Applying Theorem \ref{pro:FLG} on the next variable $x_2$ on $\langle F^3_{\overline{x_1}} \cup L_{\overline{x_1}} F^2_{\overline{x_1}} \rangle$ we get 
$$\langle F^3 \cup L F^2 \rangle \cap B[3,n]=\langle F^3_{\overline{x_1}} \cup L_{\overline{x_1}} F^2_{\overline{x_1}} \rangle\cap B[3,n]=\langle F^3_{\overline{x_1},\overline{x_{2}}} \cup L_{\overline{x_1},\overline{x_{2}}} F^2_{\overline{x_1},\overline{x_{2}}} \rangle.$$

Continuing this way for the rest of the variables to be eliminated, it follows that $\langle F^3 \cup L F^2 \rangle \cap B[k,n]=\langle F^3_{\overline{x_1},\ldots ,\overline{x_{k}}} \cup L_{\overline{x_1},\ldots ,\overline{x_{k}}} F^2_{\overline{x_1},\ldots ,\overline{x_{k}}} \rangle$ as desired.

\end{proof}

\begin{algorithm}
   \caption{$eliminate{\bf A}(F^3,F^2,x_1)$}
   \label{alg:main}
    \begin{algorithmic}
    \Require{$F^3=(f^3_1,\ldots,f^3_{r_3})$ set of cubic polynomials in $B$, $F^2=(f^2_1,\ldots,f^2_{r_2})$ set of quadratic polynomials in $B$, $x_1$ variable to be eliminated from $F^3$ and $F^2$}
    \Ensure{Set $F^3_{\overline{x_1}}$ of cubic polynomials where $x_1\not\in F^3_{\overline{x_1}}$ and set $F^2_{\overline{x_1}}$ of quadratic polynomials where $x_1\not\in F^2_{\overline{x_1}}$}
    \State 
    \State $F^2_{x_1}, F^2_{\overline{x_1}}\leftarrow SplitVariable(F^2,x_1)$\Comment{The $f^2_i\in F^2_{x_1}$ will have unique leading monomials containing $x_1$}
        \State $F^3\leftarrow (x_1+1)F^2_{x_1}\cup x_1F^2_{\overline{x_1}}\cup F^3$
        \State $F^3_{x_1},F^3_{\overline{x_1}}\leftarrow SplitVariable(F^3,x_1)$
        \State $F^{3,norm}_{\overline{x_1}},F^{3,norm}_{x_1}\leftarrow Normalize(F^3_{x_1},F^2_{x_1})$
        \State $F^3_{\overline{x_1}}\leftarrow R\_and\_C(F^2_{x_1},x_1)\cup F^3_{\overline{x_1}} \cup F^{3,norm}_{\overline{x_1}}$
    \State Return $F^3_{\overline{x_1}},F^2_{\overline{x_1}}$
    \end{algorithmic}
\end{algorithm}

The output comprises sets $F^2_{\overline{x_1}}$ and $F^3_{\overline{x_1}}$ of polynomials of degree $2$ and $3$ respectively. These are nontrivial polynomials of the same degree as the input polynomials and neither set contains the variable $x_1$.  The two sets satisfy $\langle F^2_{\overline{x_1}},F^3_{\overline{x_1}}\rangle\subseteq I\cap B[2,n]$. Algorithm \ref{alg:main} can be iterated as shown in Corollary \ref{pro:FLG:Severalvariables}, eliminating one variable from the system at the time, in any given order. An important note is that when starting with an MQ system (i.e. $F^3=\emptyset$) and eliminating only one variable, we do not throw away any polynomials in Alg. \ref{alg:main}.  Then we actually compute the full elimination ideal and are certain to preserve the initial solution space.

For complexity, we have ${\mathcal O}(n^3)$ and ${\mathcal O}(n^2)$ monomials and polynomials in $F^3_{\overline{x_1}}$ and $F^2_{\overline{x_1}}$, respectively.  If there were more polynomials than monomials we could solve the system by re-linearization.  Hence the space complexity of the algorithm is storing ${\mathcal O}(n^6)$ monomials which in practice is storing ${\mathcal O}(n^6)$ bits.

The time complexity for normalization can be estimated as follows: There are at most $n$ different $f^2_i$ in $F_{x_1}$.  For each of the ${\mathcal O}(n^3)$ polynomials in $F^3$, there may be $n$ monomials containing the leading term of $f^2_i$.  Each of these needs to be cancelled by adding a multiple of $f^2_i$, costing ${\mathcal O}(n^3)$ bit operations.  The total worst-case complexity for the normalization step is then ${\mathcal O}(n\times n^3\times n\times n^3)={\mathcal O}(n^8)$ bit operations.

The time complexity for computing resultants and coefficient constraints can be estimated similarly to also be ${\mathcal O}(n^8)$.

The time complexity for $SplitVariable()$ can be estimated as follows.  In the worst case, we have input size ${\mathcal O}(n^3)$ in both polynomials and monomials, so the matrices constructed are of size ${\mathcal O}(n^3)\times{\mathcal O}(n^3)$.  In the Gaussian reduction we need to create $0$'s under leading $1$'s in $\mathcal{O}(n^2)$ columns (those corresponding to monomials $x_1x_ax_b$), and this costs $\mathcal{O}(n^3\times n^3\times n^2)=\mathcal{O}(n^8)$ bit operations.  Hence the total time complexity for Algorithm  \ref{alg:main} is ${\mathcal O}(n^8)$ bit operations.

\subsection{Extensions of L-Elim{\bf A}(): L-Elim{\bf B}()}

In this section we improve $L-Elim{\bf A}()$ and $eliminate{\bf A}()$ by adding the heuristic procedure \noindent {\bf B}, in the following form:

\medskip

\noindent {\bf B.} It is conceivable that the vector space $\langle F^3 \cup L F^2 \rangle$ contains more relations of degree $\leq 2$, beyond the ones in $F^2$. Hence we may improve on procedure \noindent {\bf A}, by adding a search for more quadratic relations. This can be done by ordering the monomials such that the degree $3$ monomials are bigger than the degree $2$ monomials, and then split the system into degree $2$ and $3$ polynomials by calling $SplitDeg2/3()$.

Let $F^{2,(1)} = F^2$, and compute $\langle F^3 \cup L F^{2,(1)} \rangle$.
Calling $SplitDeg2/3(F^3 \cup L F^{2,(1)})$ returns new sets $F^3$ and $F^{2,(2)}$. If $\langle F^{2,(2)} \rangle$ is strictly larger than the space $\langle F^{2,(1)} \rangle$, we continue to compute $splitDeg2/3(F^3 \cup LF^{2,(2)})$, repeating the process.

We continue such computations until $\langle F^{2,(i)} \rangle = \langle F^{2,(i-1)} \rangle$ for some $i$.  Setting $F^2 := F^{2,(i)}$ we can then continue with \noindent {\bf A}, which is the elimination step. In the case that $SplitDeg2/3()$ on $F^3 \cup L F^{2,(1)}$ does not yield any new quadratic polynomials, we proceed directly to the elimination step.

\begin{Remark}{\bf The advantage of adding B:} Finding new quadratic polynomials allows to "compute with monomials of degree $\geq 4$ by computing with monomials of degree $\leq 3$". More precisely, suppose we find a new quadratic polynomial $h$, so:
\[ h = f^3 + \sum_i l_i f^2_i \]
where $f^3 \in F^3$ and the $l_i \in L$ and $f^2_i \in F^2$, and $h$ is not in $\langle F^2\rangle$.  Then $h$ can again be multiplied with a linear polynomial and added to $\langle F^3\cup LF^2\rangle$ to form

$$f^{3\prime} + lh + \sum_i l_i^\prime f^2_i  \label{eq:flh} \\
 = f^{3\prime} + lf^3 + \sum_i ll_if^2_i + \sum_i l_i^\prime f^2_i$$

The terms in $lf^3$ and $ll_if^2_i$ are generally of degree $4$, but they cancel so in reality we only work with polynomials of degree $\leq 3$. 
\end{Remark}

In Algorithm \ref{alg:elimv2x1} we present $L-Elim{\bf B}()$ which is the extended version of $L-Elim{\bf A}()$. The only difference between these two algorithms is the addition of \noindent {\bf B.}, which \emph{potentially} increases the set $F^2$.

\begin{algorithm}
    \caption{$L-Elim{\bf B}(F^3,F^2,x_1)$}
    \label{alg:elimv2x1}
    \begin{algorithmic}
    \Require{$F^3=(f^3_1,\ldots,f^3_{r_3})$ set of cubic polynomials in $B[1,n]$, $F^{2,(1)}=F^2=(f^2_1,\ldots,f^2_{r_2})$ set of quadratic polynomials in $B[1,n]$, $L=\{1, x_1 ,\ldots , x_n \}$ and $x_1$ the variable to be eliminated from $F^3$ and $F^2$}
    \Ensure{Set $F^3_{\overline{x_1}}$ of cubic polynomials and set $F^2_{\overline{x_1}}$ of quadratic polynomials, where $x_1\not\in F^2_{\overline{x_1}} \cup F^3_{\overline{x_1}}$}
        \State 
        \State $F^*\leftarrow F^3\cup L\cdot F^{2,(1)}$ \Comment{procedure {\bf A}}
        \State $F^{2,(2)}, F^3 \leftarrow SplitDeg2/3(F^*)$
        \State $i=2$
        \While{$\langle F^{2,(i)} \rangle \neq \langle F^{2,(i-1)} \rangle $} \Comment{procedure {\bf B}}
            \State $F^*\leftarrow F^3\cup L\cdot F^{2,(i)}$
            \State $F^{2,(i+1)}, F^3 \leftarrow SplitDeg2/3(F^*)$
            \State $i=i+1$
        \EndWhile
        \State $F^2=F^{2,(i)}$
        \State $F^2_{x_1}, F^2_{\overline{x_1}} \leftarrow SplitVariable (F^2,x_1)$ \Comment{procedure {\bf A}}
        \State $F^3_{x_1}, F^3_{\overline{x_1}} \leftarrow SplitVariable (F^3,x_1)$ 
       \State Return $F^3_{\overline{x_1}} F^2_{\overline{x_1}}$
    \end{algorithmic}
\end{algorithm}

\begin{Remark}{\bf Computational complexity:}
Part \noindent {\bf B} includes a while loop where we have to call $SplitDeg2/3()$ every time.  A naive implementation would require an additional $\mathcal{O}(n^9)$ operations per loop iteration. However, since $SplitDeg2/3()$ in the loop is called on the argument $F^*$ which may change only by a little for each loop iteration, it is possible that more efficient algorithms exist.
\end{Remark}

\subsection{Extensions of main elimination algorithm eliminate{\bf A}(): eliminate{\bf B}()}

In a similar manner as in the previous subsection, we can also extend $eliminate{\bf A}()$ to $eliminate{\bf B}()$.  Instead of using \textbf{B} on $\langle F^3\cup LF^2\rangle $, we use \textbf{B} on the set $F^{3,norm}_{x_1}$ which is the output of the normalization step, and the set $F^3_{\overline{x_1}}:= F^3_{\overline{x_1}}\cup F^{3,norm}_{\overline{x_1}}\cup R$ (See $eliminate{\bf A}()$ and equation (\ref{newF2F3})).

Let $F^2_{x_1}:=F^{2,(1)}_{x_1}$ and $F^2_{\overline{x_1}}:=F^{2,(1)}_{\overline{x_1}}$. Calling $SplitDeg2/3()$ on $F^{3,norm}_{x_1}$ and $F^3_{\overline{x_1}}$, returns the sets $F^{3}_{x_1}, F^{2,(2)}_{x_1}, F^3_{\overline{x_1}}$, and $F^{2,(2)}_{\overline{x_1}}$.  If either of the spaces $\langle F^{2,(2)}_{x_1} \rangle $, $\langle F^{2,(2)}_{\overline{x_1}} \rangle$ are strictly larger than the spaces $\langle F^{2,(1)}_{x_1}\rangle$ and $\langle F^{2,(1)}_{\overline{x_1}}\rangle$, we continue to compute normal forms, resultants and coefficient constraints as in $eliminate{\bf A}()$, repeating the process.

We continue these computations until $\langle F^{2,(i)}_{x_1}\rangle= \langle F^{2,(i-1)}_{x_1}\rangle$ and $\langle F^{2,(i)}_{\overline{x_1}}\rangle = \langle F^{2,(i-1)}_{\overline{x_1}}\rangle$ for some $i\geq 1$. Setting $F^2_{x_1}:=F^{2,(i)}_{x_1}$ and $F^2_{\overline{x_1}}:=F^{2,(i)}_{\overline{x_1}}$, we simply return the new sets $F^3_{\overline{x_1}}$ and $F^2_{\overline{x_1}}$.

In the case that $SplitDeg2/3()$ on $F^3_{x_1},F^3_{\overline{x_1}}$ does not yield any new quadratic polynomials, we proceed directly to $SplitVariable()$ on $F^3_{x_1}$, and return $F^3_{\overline{x_1}}$ and $F^2_{\overline{x_1}}$ as before. $eliminate{\bf B}()$ is described in Algorithm \ref{alg:mainv2}. Again the only difference between $eliminate{\bf A}()$ and $eliminate{\bf B}()$, is the partial addition of \noindent {\bf B}.

\begin{algorithm}
   \caption{$eliminate{\bf B}(F^3,F^2,x_1)$}
   \label{alg:mainv2}
    \begin{algorithmic}
    \Require{$F^3=(f^3_1,\ldots,f^3_{r_3})$ set of cubic polynomials in $B$, $F^{2,(1)}=(f^2_1,\ldots,f^2_{r_2})$ set of quadratic polynomials in $B$, $x_1$ variable to be eliminated from $F^3$ and $F^2$}
    \Ensure{Set $F^3_{\overline{x_1}}$ of cubic polynomials where $x_1\not\in F^3_{\overline{x_1}}$ and set $F^2_{\overline{x_1}}$ of quadratic polynomials where $x_1\not\in F^2_{\overline{x_1}}$}
    \State 
    \State $F^{2,(1)}_{x_1}, F^{2,(1)}_{\overline{x_1}}\leftarrow SplitVariable(F^{2,(1)},x_1)$\Comment{The $f^2_i\in F^{2,(1)}_{x_1}$ will have unique leading monomials containing $x_1$}
    \State $i=1$
    \While{$\langle F^{2,(i)}_{x_1} \rangle \neq \langle F^{2,(i-1)}_{x_1} \rangle $ or $\langle F^{2,(i)}_{\overline{x_1}} \rangle \neq \langle F^{2,(i-1)}_{\overline{x_1}} \rangle $} \Comment{procedure {\bf B}}
        \State $F^3\leftarrow (x_1+1)F^2_{x_1}\cup x_1F^2_{\overline{x_1}}\cup F^3$
        \State $F^{3,norm}_{x_1},F^{3,norm}_{\overline{x_1}}\leftarrow Normalize(F^3,F^2_{x_1})$
        \State $F^3_{\overline{x_1}}\leftarrow R\_and\_C(F^2_{x_1},x_1)$
        \State $F^{2,(i+1)}, F^3 \leftarrow SplitDeg2/3(F^{3,norm}_{x_1}\cup F^{3,norm}_{\overline{x_1}}\cup F^3_{\overline{x_1}})$
        \State $F^{2,(i+1)}_{x_1}, F^{2,(i+1)}_{\overline{x_1}}\leftarrow SplitVariable(F^{2,(i+1)},x_1)$
        \State $i=i+1$
    \EndWhile
    \State $F^3_{x_1}, F^3_{\overline{x_1}}\leftarrow SplitVariable(F^3,x_1)$
    \State $F^2_{\overline{x_1}}=F^{2,(i)}_{\overline{x_1}}$ %, $F^2_{x_1}=F^{2,(i)}_{x_1}$
    \State Return $F^3_{\overline{x_1}},F^2_{\overline{x_1}}$
    \end{algorithmic}
\end{algorithm}

The following summarizes the constructions we have done in this section.

\begin{enumerate}
    \item We can eliminate variables in two different ways: By either $L-Elim{\bf A}()$, or by $eliminate{\bf A}()$. The latter algorithm has significantly lower complexity than the first, since we avoid to multiply with \emph{all} variables in $L$.
    \item We can produce extra polynomials of degree $\leq 2$ by adding {\bf B} to $L-Elim{\bf A}()$, giving the algorithm $L-Elim{\bf B}()$.
    \item We can produce extra polynomials of degree $\leq 2$ by adding a partial version of {\bf B} to $eliminate{\bf A}()$, yielding $eliminate{\bf B}()$, with significantly lower complexity than $L-Elim{\bf B}()$.
\end{enumerate}

\subsection{Information loss} \label{subsec:enformation}
The information theoretic concepts defined in this subsection mostly follow standard notation, cf. for example \cite{Cover2006}, except for $i()$ of equation (\ref{defeq:i}) which is adapted for our purposes. Let $X$ be a discrete random variable that takes values $x_1,\ldots,x_M$ with probabilities $p_i = P(X=x_i), i=1,\ldots,M$. The  (binary) \emph{entropy} of $X$  is defined as 
\begin{equation}
    H(X) = - \sum_{i=1}^M p_i \log_2 p_i.
\end{equation}
If $X$ is uniformly distributed, \emph{i. e.} $p_i = 1/M, i=1,\ldots,M$, then $H(X) = \log_2 M$. Given $X$ and another random variable $Y$ that assumes values $y_1,\ldots,y_{M'}$,  the conditional entropy of $X$ given that we observe that $Y$ takes a specific value $y$ is 
\begin{equation}
    H(X|Y=y) = - \sum_{i=1}^M P(X=x_i|Y=y) \log_2 P(X=x_i|Y=y),
\end{equation}
and the information that we get about $X$ by observing $Y=y$ is $H(X)-H(X|Y=y)$.

The application of the concept of entropy in the context of this paper is as follows. A cryptanalyst wishes to recover a secret $k$-bit key $K$. A priori, $K$ will be assumed to be drawn uniformly from the set of all keys, so $H(K)=k$. A set $F$ of equations that $K$ must satisfy will reduce the entropy of $K$ if not all possible key values satisfy all equations in $F$. Hence $F$ contains \emph{information} about the secret key $K$ and we define
\begin{equation}
\label{defeq:i}
i(F) = \mbox{Number of key bits} - \log_2 ( \mbox{Number of key values that satisfy } F).
\end{equation}
The iterative elimination algorithms described in this section produce sequences of equation sets $F_0, F_1, F_2,\ldots$, where $F_j =  F^2_{\overline{x_1},\ldots,\overline{x_j}} \cup F^3_{\overline{x_1},\ldots,\overline{x_j}}$  denotes the set of equations contained after eliminating $j$ variables. % (along a specific trajectory of eliminated variables.) 
From the point of view of a cryptanalyst, the function $i(F_j)$ should remain high (and close to the number of key bits) for as long as possible. On the other hand, as an easy consequence of  information theory's \emph{data processing lemma}, the sequence $i(F_0),i(F_1),i(F_2),\ldots$ is non-increasing and, since high degree polynomials are discarded from equation sets to contain the complexity, it is likely that the sequence will be decreasing at some point. From the point of view of the cipher designer,  the function $i(F_j)$ should drop rapidly with $j$. Keeping track of the development of the sequence $i(F_0),i(F_1),i(F_2),\ldots$ is of interest and it will be studied in the next section.

\section{Experimental Results}
\label{sec:experiments}

We have implemented $L-Elim{\bf B}()$ and $eliminate{\bf B}()$ and done some experiments to see how they perform in practice.  In this section we report on these experiments.

\subsection{Reduced LowMC cipher}

LowMC is a family of block ciphers proposed by Martin Albrecht et al. \cite{LowMC}.  The cipher family is designed to minimize the number of AND-gates in the critical path of an encryption, while still being secure.  The cipher itself is a normal SPN network, with each round consisting of an S-box layer, an affine transformation of the cipher block and addition with a round key.  All round keys are produced as affine transformations of the user-selected key.

Two features of the LowMC ciphers are interesting with respect to algebraic cryptanalysis.  First, the S-box used is as small as possible without having linear relations among the input and output bits.  LowMC uses a $3\times 3$ S-box, where the ANF of each output bit only contains one multiplication of input bits, making the three output polynomials of the S-box quadratic.  We can search for other quadratic relations in the six input/output variables, and we then find 14 linearly independent quadratic polynomials.  

Second, the S-boxes in one round do not cover the whole state, so a part of the cipher block is not affected by the S-box layer.  The number of S-boxes to use in each round is a parameter that varies within the cipher family, and some variants are proposed with only one S-box per round.

The cipher parameters we have used for the reduced LowMC version of our experiments are:
\begin{itemize}
    \item Block size: 24 bits
    \item Key size: 32 bits
    \item 1 S-box per round
    \item 12 or 13 rounds
\end{itemize}

As will become clear below, the number of rounds is on the border of when $L-Elim{\bf B}()$ and $eliminate{\bf B}()$ are successful in breaking the reduced cipher.

\subsubsection{Constructing equation system.}

The attack is a known plaintext attack, where we assume we are given a plaintext/ciphertext pair and the task is to find the unknown key.  We use the 14 quadratic polynomials describing the S-box as the base equations.  The bits in the unknown key are assigned as the variables $x_0,\ldots,x_{31}$, and the output bits from each S-box used in the cipher are the variables $x_{32},\ldots$.  All other operations in LowMC are linear, so the input and output bits of every S-box can be written as a linear combination of the variables defined and the constants from the plaintext.

Inserting the actual linear combination for each input/output bit of the S-box in one round will produce $14r$ equations in total.  These equations describe a LowMC encryption over $r$ rounds.  The initial number of variables is $32+3r$, but this can be reduced by using the known ciphertext.  The bits of the cipher block output from the last round are linear combinations of variables.  These linear combinations are set to be equal to the known ciphertext bits, giving 24 linear equations that can be used to eliminate 24 variables by direct substitution.  After this the final number of variables is $8+3r$.  See Fig. \ref{fig:LowMC} for the equation setup.

\begin{figure}
    \centering
    \includegraphics[width=14cm]{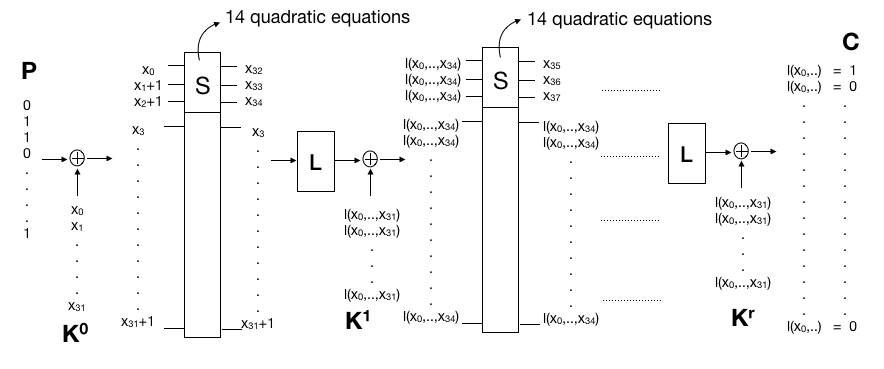}
    \caption{Setup of equation system representing reduced LowMC.  All $l(.)$'s only indicate some linear combination, and are not equal.}
    \label{fig:LowMC}
\end{figure}

\subsubsection{Experimental results.}

The goal of our experiment is to try to eliminate all the variables $x_i$ for $i\geq 32$, and find some polynomials of degree at most 3, only in variables representing the unknown user-selected key.  If we are able to find at least one polynomial only in $x_0,\ldots,x_{31}$ for one given plaintext/ciphertext pair, we can repeat for other known plaintext/ciphertext pairs and build up a set of equations that can be solved by re-linearization when the set has approximately ${32\choose 3}$ independent polynomials.

{\bf 12 rounds:} The system initially contains $44$ variables and $168$ quadratic equations.

We first use $L-Elim{\bf B}()$ to eliminate the $12$ variables with highest indices.  With this method we succeed in producing 1-2 cubic polynomial(s) only in key variables (some p/c-pairs produce 1, others produce 2 polynomials).  The memory requirement is to store the 7560 polynomials we get after multiplying the quadratic equations with all terms in $L$.  

Next we apply the $eliminate{\bf B}()$ algorithm on the same system.  Initially the set $F^2$ contains 168 polynomials and the set $F^3$ is empty.  As the algorithm proceeds, eliminating one variable at the time, the sizes of $F^3$ and $F^2$ change.  The set $F^3$ grows at first before starting to decrease before the last variables are eliminated, while the set $F^2$ decreases at a steady pace during the 12 eliminations.  The size of $F^3$ was never above 2000 polynomials, so $eliminate{\bf B}()$ has considerably less space complexity than $L-Elim{\bf B}()$.  The observed running time of the two methods were roughly the same, and $eliminate{\bf B}()$ produced the same polynomials as $L-Elim{\bf B}()$ in the end.

Finally we generate 15 different systems using different p/c-pairs, to see how many independent polynomials in $x_0,\ldots,x_{31}$ we get when collecting all outputs from the 15 systems together.  The 15 systems collectively produced 20 polynomials in only key bits, of which 16 were linearly independent.  So the hypotheses that we can produce many independent polynomials from different p/c-pairs seems to hold.

At this stage we noticed something unexpected.  After doing Gaussian elimination on the 20 polynomials to check for linear dependencies, it turned out that we produced five {\it linear} polynomials in the unknown key variables.  It therefore appears that the polynomials produced from the elimination algorithm are not completely random, and that one may need much fewer polynomials than anticipated to actually find the values of $x_0,\ldots,x_{31}$.

{\bf 13 rounds:} The initial system contains 47 variables and 182 quadratic equations.

Neither $L-Elim{\bf B}()$ nor $eliminate{\bf B}()$ were able to find any cubic polynomials in only $x_0,\ldots,x_{31}$ for any 13-round systems we tried.  So for the reduced LowMC version we used, only up to 12 rounds may be attacked using our elimination techniques and bounding the degree to at most 3.

\subsection{Toy Cipher}

For the experiments we also made a small toy cipher to do tests on.  The toy cipher has a 16-bit block and a 16-bit key, and is built as a normal SPN network.  Each round consists of an S-box layer with four $4\times 4$ S-boxes (the same S-box as used in PRINCE), followed by a linear transformation and a key addition.  The same key is used in every round.  For the elimination experiments reported here we use a 4-round version of the toy cipher.

\subsubsection{Constructing equation system.}

The equation system representing the toy cipher is constructed similarly to the reduced LowMC.  The variables in the unknown key are $x_0,\ldots,x_{15}$, and the output bits of every S-box, except for the last round, are variables $x_{16},\ldots,x_{63}$.  The inputs and outputs of every S-box can then be described as linear combinations of the variables we have defined, together with the constants in the known plaintext and ciphertext blocks.  See Fig. \ref{fig:toyCipher} for the setup of the equations.  

Each output bit of the PRINCE S-box has degree 3 when written as a polynomial of the input bits, but there exists 21 quadratic relations in input/output variables describing the S-box.  The number of quadratic equations in the 4-round toy cipher is therefore 336, in the 64 variables $x_0,\ldots,x_{63}$.

\begin{figure}
    \centering
    \includegraphics[width=14cm]{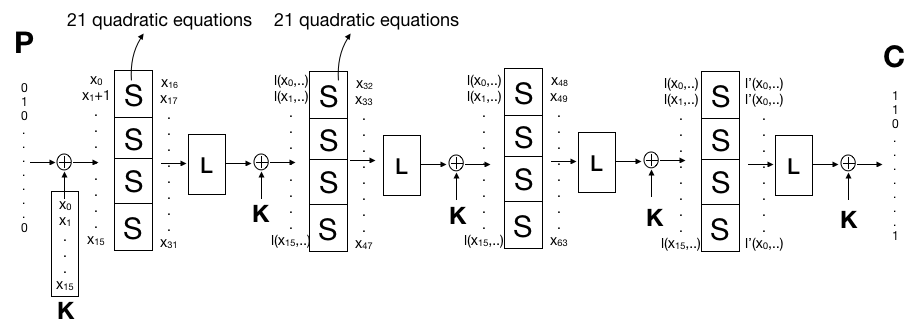}
    \caption{Setup of equation system representing 4-round toy cipher. All $l(.)$'s and $l'(.)$'s indicate some linear combination of variables.}
    \label{fig:toyCipher}
\end{figure}

\subsubsection{Experimental results.}

When trying to eliminate all non-key variables $x_{16},\ldots,x_{63}$ from the system, neither $L-Elim{\bf B}()$ nor $eliminate{\bf B}()$ were able to find any cubic polynomial in only $x_0,\ldots,x_{15}$.

We know that when running $eliminate{\bf B}()$ we will throw away polynomials giving constraints on the solution space on the way, and hence introduce false solutions.  When $F^3$ and $F^2$ become empty the whole space becomes the solution space, and we have lost all information about the possible solutions to the original equation system.  It is interesting to measure how fast the information about the solutions we seek disappear, and this is what we have investigated for the toy cipher. 

As in all algebraic cryptanalysis we are interested in finding the possible values for the secret key.  In this case this means finding the values of $x_0,\ldots,x_{15}$.  With only a 16-bit key it is possible to do exhaustive search, and check which key values that fit in any of the equation systems we get after eliminating some variables.  The procedure we used for checking if one guessed key fits in a given system is as follows:

\begin{itemize}
    \item Fix $x_0,\ldots,x_{15}$ to the guessed value in the system
    \item Do Gauss elimination on the resulting system to produce linear equations
    \item Use each linear equation found to eliminate one more variable
    \item Repeat Gauss elimination to find new linear equations and new eliminations, etc.
    \item If we find the polynomial $1$ after Gauss elimination the guessed key does not fit
    \item If all variables get eliminated without producing any $1$-polynomial, the guessed key fits
    \item If we fail to produce linear equations in the Gauss elimination, it is undecided whether the guessed key fits or not
\end{itemize}

We set up an elimination order where variables to be eliminated were distributed evenly throughout the system.  That is, we do not eliminate the second variable from an S-box before all S-boxes have at least one variable eliminated.  The exact elimination order used was 
$$x_{36},x_{24},x_{52},x_{44},x_{20},x_{56},x_{40},x_{28},x_{60},x_{32},x_{16},x_{48},x_{18},x_{50},x_{34},x_{26},$$
$$x_{58},x_{46},x_{54},x_{22},x_{62},x_{30},x_{38},x_{42},x_{47},x_{21},x_{49},x_{35},x_{29},x_{59},x_{41}.$$

After eliminating these 31 variables, all keys fit in the system we have at that point.  For each system we get along the way, we checked how many keys that fit in the given system.  This gives a measure of how much information the system has about the unknown secret key we try to find.  For a system $F$, we use $i(F)$ (\ref{defeq:i}) that says how much information the system has about the key:
$$i(F)=16-log_2(\mbox{\# of keys that fit in $F$}).$$

Denote the system we have after eliminating $v$ variables as $F_v$.  For the plaintext/ciphertext pair we used there were three keys that fit in the initial system, so we have $i(F_0)\approx 14.42$.  We know that $i(F)$ is a strictly non-increasing function for increasing $v$, because we can only lose information during elimination.  Put another way, if the key $K$ fits in $F_v$, $K$ will also fit in $F_w$ for $w>v$.  It is interesting to see what the rate of information loss is during elimination.  Is the information loss gradual, or do we lose all information more suddenly?  In Fig. \ref{fig:plot} we have plotted the graph for $i(F_v)$ for $0\leq v\leq 31$.

\begin{figure}
    \centering
    \includegraphics[width=14cm]{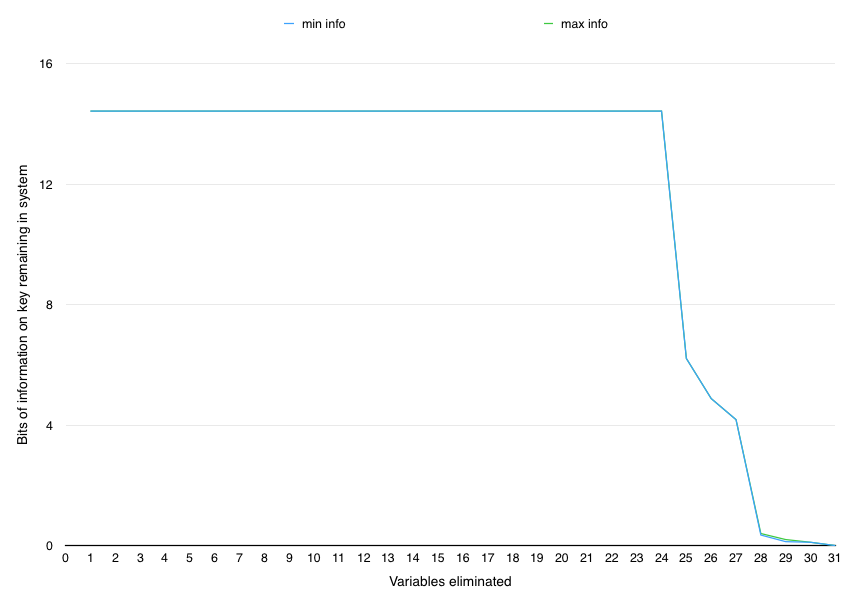}
    \caption{Information loss when eliminating variables from 4-round toy cipher.}
    \label{fig:plot}
\end{figure}

As we can see in Fig. \ref{fig:plot}, we can eliminate up to 24 of the 48 non-key variables in the system without losing any information on the possible keys.  The three keys that fit in the original system are still the only ones that fit in $F_{24}$.  After that, all information on possible keys is lost rather quickly, and $i(F_{31})=0$.  Only in $F_{27}$ and $F_{28}$ did we run into some cases where it could not be decided whether a guessed key fits or not.  This is barely visible in Fig. \ref{fig:plot}, where there is a tiny area where the true values of $i(F_{27})$ and $i(F_{28})$ may lie.

We find this behavior interesting and a source for further study.  We can look at it this way: It is possible to describe a cipher by quadratic equations in $k$ key variables and $n-k$ non-key variables (i.e. constructed as in Figs. \ref{fig:LowMC} and \ref{fig:toyCipher}).  Our experiment indicates that (at least sometimes) one can create a cubic equation system, with the same information on the key, with only $k+(n-k)/2$ variables.  In other words, there is a trade-off between degree and number of variables needed to describe a cipher.  For the toy cipher, increasing the degree by one allows to cut the number of non-key variables in half to describe the same cipher.  

\section{Conclusions}

In this paper we proposed two new algorithms for performing elimination of variables from systems of Boolean equations: %In particular we develop a mathematical framework for elimination, which shows that we can eliminate variables in two different ways: $1)$ By 
$L-Elim{\bf *}()$ which is essentially Gaussian elimination, and $eliminate{\bf *}()$ which is more efficient and when suitably extended also more effective. We applied these algorithms in a known plaintext attack to two reduced versions of the LowMC cipher: $12$ and $13$ rounds with $24$ bits block and $32$ bits key. For the $12$-round version the algorithms produces polynomials of degree $3$ in only key variables, while in the $13$-round example the algorithms fail to find any polynomials of degree $3$ in only key variables. 

We also applied the algorithms to a toy cipher for performing tests, where the proposed algorithms fails to find any polynomials of degree $3$ in only key variables. Instead we extend the experiments by measuring how much information we lose about the key during elimination. Surprisingly, the experiments show that we can eliminate many auxiliary variables from the system of equations, without losing any information about the key. Another result of the experiments is that we lose information about the key rather quickly after a certain point in the elimination process. We conclude that there is a lot of future work to be done in this direction.

\section*{Appendix A: Monomial orders and splitting algorithms}
\label{app:A}
%\subsection{Two different monomial orders achieving different goals in elimination of variables}\label{2monomialorders}

Consider the vector space $\langle F\rangle = \langle F^3\cup LF^2\rangle$ which is generated by the set of Boolean polynomials $F^3,F^2$.  We can perform Gaussian reduction on this vector space with two different orders.  In the Gaussian elimination we order the monomials such that the largest monomials are eliminated first.

\medskip

\noindent {\bf A.} The monomial order where $x_1$-monomials are largest: $\langle F\rangle$ can be realised as a matrix $A$.  Each row of $A$ corresponds to one polynomial in $F^3\cup LF^2$, and each column corresponds to one monomial $m$. Moreover, the entry $A[i,j]$ corresponds to the coefficient of the $j$'th monomial in the $i$'th polynomial.  When $x_1$-monomials are the largest, we consider the leftmost columns of $A$ to correspond to all monomials containing $x_1$. Note that for the matrix $A$, we write $A[i\ldots j]$ to indicate the submatrix consisting of rows $i$ through $j$ of $A$.  With a slight abuse of notation, we write $x_1\in m, x_1\in f$ or $x_1\in G$ to indicate that $x_1$ occurs in monomial $m$, polynomial $f$ or polynomial set $G$.

When performing Gaussian elimination on $A$ with this order, we can create polynomials in the span of $F^3\cup LF^2$ that have $0$'s in the leftmost columns. If there are enough polynomials in $F^3\cup LF^2$, the lower rows of $A$ will then give a non-empty set of polynomials $F^3_{\overline{x_1}} \cup LF^2_{\overline{x_1}}$ that do not contain the $x_1$-variable. Note that the new set $F^3_{\overline{x_1}}\cup LF^2_{\overline{x_1}}\supseteq F^3\cup LF^2\cap\Bton$.

\medskip

\noindent {\bf B.} The order where higher-degree monomials are larger: It is conceivable that $\langle F\rangle=\langle F^3 \cup LF^2 \rangle$ contains more quadratic polynomials than just the ones in $F^2$. These can be found if we order the monomials such that the degree $3$ monomials are bigger than degree $2$ monomials.  We can then use Gaussian elimination on the matrix $A$ representing $\langle F\rangle=\langle F^3\cup LF^2\rangle$ to eliminate monomials of degree $3$ and possibly produce more quadratic equations than there are originally in $F^2$.

The algorithm for splitting polynomial sets into those containing $x_1$ and those which do not contain $x_1$ is given in Algorithm \ref{alg:splitx1} below. The algorithm for splitting a set of degree $3$ polynomials into degree $2$ and $3$ polynomials is given in Algorithm \ref{alg:split23} below.
We are going to use these orders in section \ref{Algorithms} as building blocks for finding more quadratic and cubic polynomials when developing the elimination algorithms.

\begin{algorithm}
    \caption{$SplitVariable(F,x_1)$}
    \label{alg:splitx1}
    \begin{algorithmic}
      \Require{$F=(f_1,\ldots,f_m)$ set of polynomials of degree $\leq 3$ in $B[1,n]$}
      \Ensure{Sets $F_{x_1}$ and $F_{\overline{x_1}}$ of polynomials such that $\langle F\rangle=\langle F_{x_1}\cup F_{\overline{x_1}}\rangle, x_1\in F_{x_1}$ and $x_1\not\in F_{\overline{x_1}}$}
      \State 
      \State $\mathbf{m}=(m_1,\ldots,m_c,m_{c+1},\ldots,m_t)\leftarrow$ monomials occurring in $F$ where $x_1\in m_i$ for $1\leq i\leq c$ and $x_1\not\in m_i$ for $i>c$.
      \State $A\leftarrow m\times t$ matrix where coefficient of $m_j$ in $f_i$ is entry $A[i,j]$
      \State Row-reduce $A$ such that leading 1's in rows $i\leq r$ are in columns $j\leq c$ and leading 1's in rows $i>r$ are in columns $j>c$.
      \State $F_{x_1}=A[1\ldots r]\mathbf{m}^T$
      \State $F_{\overline{x_1}}=A[r+1\ldots t]\mathbf{m}^T$
      \State Return $F_{x_1}, F_{\overline{x_1}}$
    \end{algorithmic}
\end{algorithm}

\begin{algorithm}
   \caption{$SplitDeg2/3(F)$}
   \label{alg:split23}
    \begin{algorithmic}
      \Require{$F=(f_1,\ldots ,f_m)\subseteq B[1,n]$ is set of polynomials of degree $\leq 3$.}
      \Ensure{Sets $F^2$ of quadratic polynomials and $F^3$ of cubic polynomials such that $\langle F\rangle=\langle F^2\cup F^3\rangle$.}
      \State 
      \State $\mathbf{m}=(m_1,\ldots,m_c,m_{c+1},\ldots,m_t)\leftarrow$ monomials occurring in $F$ where $deg(m_i)=3$ for $1\leq i\leq c$ and $deg(m_i)\leq 2$ for $i>c$.
      \State $A\leftarrow m\times t$ matrix where coefficient of $m_j$ in $f_i$ is entry $A[i,j]$
      \State Row-reduce $A$ such that leading 1's in rows $i\leq r$ are in columns $j\leq c$ and leading 1's in rows $i>r$ are in columns $j>c$.
      \State $F^2=A[1\ldots r]\mathbf{m}^T$
      \State $F^3=A[r+1\ldots t]\mathbf{m}^T$
      \State Return $F^2, F^3$
    \end{algorithmic}
\end{algorithm}

\section*{Appendix B: Normalizing Cubics with Respect to Quadratics}
%\subsection{}\label{normalisation}

In this appendix we present the concept of  \emph{normalization}. This procedure eliminates particular monomials containing the targeted variable $x_1$ from a set of cubic polynomials using a set of quadratic polynomials as a basis. %The effect of this procedure is that there is a rather large set of monomials containing $x_1$ that cannot appear in the output of cubic polynomials at the end, and hence optimizes the elimination procedure when applied to $\langle F_{x_1}^3\cup LF_{x_1}^2\rangle$. 
This is a heuristic procedure that attempts to remove monomials containing a variable $x_1$ from a set of polynomials. Experiments indicate that this normalization usually has a beneficial effect on both efficiency and information preservation. Moreover, the procedure is a technical requirement for the proof of Theorem~\ref{pro:FLG}. Before giving the algorithm, we develop a mathematical foundation around the process of normalization.

\medskip

Since we in this paper are considering the sets $F_{x_1}^2$ and $F_{x_1}^3$, we normalize the polynomials in $F_{x_1}^3$ with respect to the set $F_{x_1}^2$ and the variable $x_1$. With the orders on the monomials introduced in Section \ref{Notation}, it follows that any non-zero Boolean polynomial $f^3\in B[1,n]$ of degree $3$ has a {\it leading term}. This is the largest monomial in $f$ with respect to the given order. For a given set $F^2 = \{ f^2_1, \ldots, f^2_{r_2}\}$ of quadratic polynomials with distinct leading terms, the polynomial $f^3$ is in {\it normal form} with respect to the set $F_{x_1}^2$, if no monomial in $f^3$ is divisible by the leading term of any polynomial in $F_{x_1}^2$.  A polynomial $f^3$ can be brought into a normal form $f^{3,norm}$ (not in general unique) by successively subtracting multiples of the polynomials in $F_{x_1}^2$.
More specifically, we obtain $f^{3,norm}$ by the following procedure. 
Let
\[ 
f^3 = m_{f^3} + \text{ lower order terms}, \quad f^2_i = m_{f^2_i} + \text{ lower order terms},
\]
and assume that $m_{f^2_i}$ divides $m_{f^3}$. Then we can write $m_{f^3} = qm_{f^2_i}$ where $q$ is a monomial whose set of variables is disjoint from that of $m_{f^2_i}$.  We can now replace $f^3$ by $f^3+q{f^2_i}$, cancelling the term $m_{f^3}$ in the process.  Doing this successively will eventually produce the normal form of $f^3$ with respect to $f^2_i$, and performing this for all generators will eventually produce the normal form of $f^3$ with respect to the set $F_{x_1}^2$.

Note that there is a specific case that merits attention, namely when there is a polynomial $f^2_i$ in $F_{x_1}^2$ with leading term $x_1$. Then this term is the only term in $f^2_i$ involving the variable $x_1$. To distinguish this polynomial, we denote it by 

$$f^2_\infty = x_1 + h, \, \, h\in B[2,n].$$

Extra care is needed when $F_{x_1}^2$ contains $f^2_\infty$ with leading term $x_1$. The reason is that when following the procedure for making normal forms, we would remove every term in the polynomials of $F_{x_1}^3$ containing the variable $x_1$. This will also imply that we replace $f^3\in F_{x_1}^3$ by $f^3+qf^2_\infty$, where $q$ is a quadratic monomial.  Then the new $f^{3,norm}$ in general will involve terms of degree $4$, which we do not allow. However, we may still freely use $f^2_\infty$ to remove all {\it quadratic terms} in $f^3$ containing $x_1$. Hence, when $f^2_\infty$ is found in $F_{x_1}^2$ there will be no quadratic monomials in $F_{x_1}^3$ containing $x_1$ after normalization. A normal form of $f^3$ using this procedure we call a {\it $3$-normal form}, to signify that we do not do computations with monomials of degree $\geq 4$.

The complete algorithm for producing normal forms for a set $F_{x_1}^3$ of cubic polynomials using a set $F_{x_1}^2$ of quadratic polynomials as a basis, including possibly $f^2_\infty\in G_2$, is given in Algorithm \ref{alg:norm}.

\begin{algorithm}
   \caption{$Normalize(F_{x_1}^3,F_{x_1}^2)$}
   \label{alg:norm}
    \begin{algorithmic}
      \Require{$F_{x_1}^3=(f^3_1,\ldots,f^3_{r_3})$ set of cubic polynomials in $B[1,n]$, $F_{x_1}^2=(f^2_1,\ldots,f^2_{r_2})$ set of quadratic polynomials in $B[1,n]$ with $m_{f^2}$ unique leading term (in some order) in $f^2_i$}
      \Ensure{Set $F_{\overline{x_1}}^{3,norm}$ where no monomial $m\in F_{\overline{x_1}}^{3,norm}$ contains $x_1$} and set $F_{x_1}^{3,norm}$ where each polynomial contains at least one monomial  with $x_1$
      \State 
      \State $F_{\overline{x_1}}^{3,norm},F_{x_1}^{3,norm}\leftarrow \emptyset$
      \For{$f^2\in F_{x_1}^2$}
        \State $m_{f^2}\leftarrow$ leading monomial in $f^2$
            \If{$m_{f^2}=x_1$}
                \State $d \leftarrow 2$
            \Else 
                \State $d \leftarrow 3$
            \EndIf
            %\If{$f^2\neq f^2_\infty$}
                \For{$f^3\in F_{x_1}^3$}
                    \For{all monomials $m\in f^3, \deg(m) \leq d$ }
                        \If{$m_{f^2}$ divides $m$}
                            \State $f^3\leftarrow f^3+\frac{m}{m_{f^2}}f^2$ \Comment{eliminate monomial divisible by $m_{f^2}$}
                        \EndIf
                    \EndFor
                \EndFor    
            %\EndIf
        \EndFor
        \For{$f^3\in F_{x_1}^3$}
            \If{$x_1\notin f^3$}
                \State $F_{\overline{x_1}}^{3,norm}\leftarrow F_{\overline{x_1}}^{3,norm}\cup f^3$
            \Else
                \State $F_{x_1}^{3,norm}\leftarrow F_{x_1}^{3,norm}\cup f^3$
            \EndIf
        \EndFor
    \State Return $F_{\overline{x_1}}^{3,norm},F_{x_1}^{3,norm}$ %=(f^3_1,\ldots,f^3,_{r_3})$ $F_{x_1}^{3,norm}=(f^3_1,\ldots,f^3,_{r_3})$
    \end{algorithmic}
\end{algorithm}

%\begin{Remark}
%It is important to note here that one can improve the normalization of $F$, by removing even more monomials containing $x_1$. This can be done by adding \noindent {\bf B.}, which aims to produce a larger set of quadratic polynomials $G^{(2)}$. This can easily be done by Gaussian elimination on degree 3 monomials in order to try to produce some polynomials of degree 2 using Algorithm \ref{alg:split23}. In fact, in the next section we show that this has a significant effect on elimination of variables.
%\end{Remark}

\end{document}